%% file: main.tex
\documentclass[11pt]{article}
\usepackage{amsmath,graphicx}
\usepackage{tabularx}
\usepackage{amsfonts}
\usepackage{stackrel}
\usepackage[ruled, linesnumbered]{algorithm2e}
\usepackage{amsthm}
\usepackage{bbm}
\usepackage{bm}

\usepackage{palatino}

\usepackage[usenames,dvipsnames]{xcolor}

\usepackage{tcolorbox}
\tcbuselibrary{skins,breakable}
\tcbset{enhanced jigsaw}

\definecolor{DarkRed}{rgb}{0.5,0.1,0.1}
\definecolor{DarkBlue}{rgb}{0.1,0.1,0.5}

\colorlet{YellowOrange}{RawSienna}

\usepackage{hyperref}%[backref=page]
\hypersetup{
colorlinks=true,
pdfnewwindow=true,
citecolor=ForestGreen,
linkcolor=DarkRed,
filecolor=DarkRed,
urlcolor=DarkBlue
}

\allowdisplaybreaks

\newtheorem{theorem}{Theorem}
\newtheorem{claim}{Claim}[section]
\newtheorem{lemma}[claim]{Lemma}
\newtheorem{definition}{Definition}

\newtheorem{example}{Example}

\newtheorem*{lemma*}{Lemma}

\input{macros}

\title{Sublinear Time Hypergraph Sparsification via Cut and Edge Sampling Queries}
\date{}

\newcommand*\samethanks[1][\value{footnote}]{\footnotemark[#1]}

\newcommand{\OVA}{O_{\rm value}}
\newcommand{\OED}{O_{\rm edge}}
\newcommand{\ONONE}{O_{\rm nbr}^1}
\newcommand{\ONTWO}{O_{\rm nbr}^2}

\author{Yu Chen\thanks{Department of Computer and Information Science, University of Pennsylvania. Email: {{\small {\tt \{chenyu2,sanjeev\}@cis.upenn.edu.}}}}
\and Sanjeev Khanna\samethanks
\and Ansh Nagda\thanks{University of Washington. Email: {{\small {\tt ansh@cs.washington.edu.}}}}
}

\begin{document}
\maketitle

\begin{abstract}
The problem of sparsifying a graph or a hypergraph while approximately preserving its cut structure has been extensively studied and has many applications.
In a seminal work, Benczúr and Karger (1996) showed that given any $n$-vertex undirected weighted graph $G$ and a parameter $\eps \in (0,1)$, there is a near-linear time algorithm that outputs a weighted subgraph $G'$ of $G$ of size $\tilde{O}(n/\eps^2)$ such that the weight of every cut in $G$ is preserved to within a $(1 \pm \eps)$-factor in $G'$. The graph $G'$ is referred to as a {\em $(1 \pm \eps)$-approximate cut sparsifier} of $G$. Subsequent recent work has obtained a similar result for the more general problem of hypergraph cut sparsifiers. However, all known sparsification algorithms require $\Omega(n + m)$ time where $n$ denotes the number of vertices and $m$ denotes the number of hyperedges in the hypergraph. Since $m$ can be exponentially large in $n$, a natural question is if it is possible to create a hypergraph cut sparsifier in time polynomial in $n$, {\em independent of the number of edges}. We resolve this question in the affirmative, giving the first sublinear time algorithm for this problem, given appropriate query access to the hypergraph.

Specifically, we design an algorithm that constructs a $(1 \pm \eps)$-approximate cut sparsifier of a hypergraph $H(V,E)$ in polynomial time in $n$, independent of the number of hyperedges, when given access to the hypergraph using the following two queries:
    \begin{enumerate}
        \item given any cut $(S,\bar{S})$, return the size $\card{\delta_{E}(S)}$ ({\em cut value queries}); and
        \item given any cut $(S,\bar{S})$, return a uniformly at random edge crossing the cut ({\em cut edge sample queries}).
    \end{enumerate}
   Our algorithm outputs a sparsifier with $\tilde{O}(n/\eps^2)$ edges, which is essentially optimal. 
We then extend our results to show that cut value and cut edge sample queries can also be used to construct hypergraph {\em spectral sparsifiers} in $\poly(n)$ time, independent of the number of hyperedges.

We complement the algorithmic results above by showing that any algorithm that has access to only one of the above two types of queries can not give a hypergraph cut sparsifier in time that is polynomial in $n$. Finally, we show that our algorithmic results also hold if we replace the cut edge sample queries with a {\em pair neighbor sample query} that for any pair of vertices, returns a random edge incident on them. In contrast, we show that having access only to cut value queries and queries that return a random edge incident on a given single vertex, is not sufficient.
    
\end{abstract}

\input{intro}

\input{prelim}

\input{algorithm}

\input{neighbor}

\input{spectral}

\input{lowerbound}

\input{weighted.tex}

\input{conclusions}

\section*{Acknowledgements}
We thanks the anonymous referees for their valuable comments on an earlier version of this paper. This research was partially supported by NSF awards CCF-1763514, CCF-1617851, CCF-1934876, and CCF-2008305.

\bibliographystyle{plainurl}% the mandatory bibstyle

\bibliography{general}

\end{document}

%% file: macros.tex
\newcommand{\toShrink}{-.20cm}
\newcommand{\toShrinkEnu}{-.2cm}

\newcommand{\tvd}[2]{\ensuremath{\norm{#1 - #2}_{tvd}}}
\renewcommand{\tvd}[2]{\ensuremath{\Delta_{\textnormal{\texttt{TV}}}(#1,#2)}}

\newcommand{\eps}{\ensuremath{\varepsilon}}

\newcommand{\paren}[1]{\ensuremath{\left(#1\right)}\xspace}
\newcommand{\card}[1]{\left\vert{#1}\right\vert}

\newcommand{\prob}[1]{\Pr\paren{#1}}

\newcommand{\poly}{\mbox{\rm poly}}

\DeclareMathOperator*{\Prob}{\ensuremath{\textnormal{Pr}}}
\renewcommand{\Pr}{\Prob}

\newenvironment{tbox}{\begin{tcolorbox}[
		enlarge top by=5pt,
		enlarge bottom by=5pt,
		 breakable,
		 boxsep=0pt,
                  left=4pt,
                  right=4pt,
                  top=10pt,
                  boxrule=1pt,toprule=1pt,
                  colback=white,
                  arc=-1pt,
                  ]%%
	}
{\end{tcolorbox}}

%% file: intro.tex
\section{Introduction}

In many applications, the underlying graphs are too large to fit in the main memory, and one typically builds a compressed representation that preserves relevant properties of the graph. 
Cuts in graphs are a fundamental object of study, and play a central role in the study of graph algorithms. Consequently, the problem of {\em sparsifying} a graph while approximately preserving its cut structure has been extensively studied (see, for instance,~\cite{Karger:1993:GMR:313559.313605, benczur1996approximating,Karger99, SpielmanT04, AhnG09, AhnGM12b,GoelKP12, BatsonSS12, AhnGM13, LeeS17, KapralovLMMS17, BansalST19, KapralovMMMNST20}, and references therein).
A cut-preserving sparsifier not only reduces the space requirement for any computation, but it can also reduce the time complexity of solving many fundamental cut, flow, and matching problems as one can now run the algorithms on the sparsifier which may contain far fewer edges.
In a seminal work, Benczúr and Karger~\cite{benczur1996approximating} showed that given any $n$-vertex undirected weighted graph $G$ and a parameter $\eps \in (0,1)$, there is a near-linear time algorithm that outputs a weighted subgraph $G'$ of $G$ of size $\tilde{O}(n/\eps^2)$ such that the weight of every cut in $G$ is preserved to within a multiplicative $(1 \pm \eps)$-factor in $G'$. The graph $G'$ is referred to as the {\em $(1 \pm \eps)$-approximate cut sparsifier} of $G$.

In this work, we consider the problem of cut sparsification for hypergraphs. A hypergraph $H(V,E)$ consists of a vertex set $V$ and a set $E$ of hyperedges where each edge $e \in E$ is a subset of vertices. The {\em }rank of a hypergraph is the size of the largest edge in the hypergraph, that is, $\max_{e \in E}\card{e}$. Hypergraphs are a natural generalization of graphs and many applications require estimating cuts in hypergraphs (see, for instance, ~\cite{CatalyurekA99,CatalyurekBDBHR09,HuangLM09,YamaguchiOTI15}).
Note that unlike graphs, an $n$-vertex hypergraph may contain exponentially many (in $n$) hyperedges. It is thus natural to ask if cut-preserving sparsifiers in the spirit of graph sparsifiers can also be created for hypergraphs as this would allow algorithmic applications to work with hypergraphs whose size is polynomially bounded in $n$.  

Kogan and Krauthgamer~\cite{kogan2015sketching} initiated a study of this basic question and showed that given any weighted hypergraph $H$, there is an $O(mn^2)$ time algorithm to find a $(1 \pm \eps)$-approximate cut sparsifier of $H$ of size $\tilde{O}(\frac{nr}{\eps^2})$ where $r$ denotes the rank of the hypergraph. Similar to the case of graphs, the {\em size} of a hypergraph sparsifier refers to the number of edges in the sparsifier. Since $r$ can be as large as $n$, in general, this gives a hypergraph cut sparsifier of size $\tilde{O}(n^2/\eps^2)$, which is a factor of $n$ larger than the Benczúr-Karger bound for graphs. 
Chekuri and Xu~\cite{Chekuri018} designed a more efficient algorithm for building a hypergraph sparsifier. They gave a near-linear time algorithm in the total representation size (sum of the sizes of all hyperedges) to construct a hypergraph sparsifier of size $\tilde{O}(nr^2/\eps^2)$ in hypergraphs of rank $r$, thus speeding up the run-time obtained in the work of Kogan and Krauthgamer~\cite{kogan2015sketching} by at least a factor of $n$, but at the expense of an increased sparsifier size. Until recently, it was an open question if the Benczúr-Karger bound is also achievable on hypergraphs, that is, do there exist hypergraph sparsifiers with $\tilde{O}(n/\eps^2)$ edges. In a very recent work~\cite{CKN20}, we were able to resolve this question in the affirmative by giving a $\tilde{O}(mn + n^{10}/\eps^7)$ time algorithm for creating hypergraph sparsifiers of size $\tilde{O}(n/\eps^2)$.

All known results for creating $\poly(n)$ size hypergraph sparsifiers have at least one thing in common -- the running time of these algorithms has at least a linear dependence on both $n$ and $m$. All known algorithms are essentially based on sampling edges in proportion to their importance, and they primarily differ in how the importance of an edge is defined and computed.
A linear dependence on $n$ is unavoidable since the hypergraph size itself is $\Omega(n)$. However, since the number of hyperedges can be exponential in $n$, even a linear dependence on $m$ means that the running time of a sparsification algorithm can be exponentially large in $n$ in the worst-case. This motivates the following natural question: is there an algorithm for building a hypergraph sparsifier that runs in time that is polynomial in $n$? In other words, is there a sublinear time algorithm for creating hypergraph sparsifiers? 

In order to tackle this question, we need to first define our model for accessing the input hypergraph $H=(V,E)$. The most basic requirement is to have the ability to efficiently evaluate the size or weight of any cut in a given hypergraph. 
We assume here access to a {\em cut value oracle}, denoted as $\OVA$, which takes as input a cut $C = (S,\bar{S})$, returns the size of the cut $\card{\delta_H(S)}$. This is akin to the standard assumption in submodular function minimization, namely, the algorithm has an oracle access to the value of the submodular function on any set $S$ since the cut function is a submodular function. However, as it turns out, it is easy to show that the access to a cut value oracle is provably not sufficient to construct a sparsifier, regardless of the time allowed as this oracle can not differentiate between hypergraphs where all edges have size $2$ from hypergraphs where all edges have size $3$\footnote{For instance, the cut value oracle can not distinguish between a copy of $K_4$ and the hypergraph that contains all possible hyperedges of size $3$ on $4$ vertices. Note that this does not rule out the possibility of efficiently constructing a data structure/sketch that can be used to answer cut queries. Our focus in this paper, however, is on constructing sparsifiers, namely, sparse subgraphs of the original graph that preserve all cuts.}. So we also need a mechanism for accessing edges of the underlying graph. We thus introduce a second oracle, referred to as the {\em cut edge oracle}, denoted as $\OED$, which takes as input a cut $C = (S,\bar{S})$, returns a random edge crossing the cut. Given access to both these oracles, 
we are indeed able to solve the problem of hypergraph sparsification in polynomial time in $n$.

\begin{theorem} \label{thm:main}
    Suppose we are given an unweighted hypergraph $H=(V,E)$ that can be accessed using the oracles $\OVA$ and $\OED$. Then for any $0<\eps<1$, a $(1 \pm \eps)$-approximate sparsifier with $\tilde{O}(n/\eps^2)$ hyperedges can be constructed in $O(n^{10}/\eps^7)$ time, independent of the number of hyperedges.
\end{theorem}

At a high-level, graph and hypergraph sparsification algorithms work by estimating the importance of each edge in preserving cut sizes, and then sampling edges with probability proportional to their importance and assigning them an appropriately scaled weight.  
The main technical challenge in proving the above theorem is that the cut value oracle on the original graph cannot be used to estimate cut sizes in the vertex-induced subgraphs of the original graph -- a step that is implicit in determining importance of edges in preserving the cut structure. This issue does not arise in normal graphs where each edge contains $2$ vertices, and the cut value oracle on the original graph indeed suffices to recover cut values in any induced subgraph. But once we consider hypergraphs with even edges of size $3$, it is easy to show that the cut value oracle on the original graph can not distinguish between induced subgraphs that have minimum cut value $0$ and induced subgraphs where the minimum cut value is polynomially large. We refer the reader to Section~\ref{sec:OED} for a more detailed discussion of this. We get around this issue by introducing for any subset of vertices $X$, a weaker notion of {\em pseudo cut size} for approximating cut sizes in the subgraph induced by $X$. The new cut size function remains submodular, and we show that it suffices to approximate the importance of each edge to within a factor $n$ of its true importance. We then use the $\OED$ oracle to sample edges in accordance with their approximate importance. The resulting sparsifier $H'$ has $\poly(n)$ edges which we further sparsify to $\tilde{O}(n/\eps^2)$ edges in $\poly(n)$ time by applying the result of~\cite{CKN20} to $H'$.

We complement the algorithmic result above by showing that just like the oracle $\OVA$ alone is not sufficient to achieve the result above, the oracle $\OED$ alone is also not sufficient to create a $\poly(n)$ size hypergraph sparsifier in $\poly(n)$ time.

\begin{theorem}
\label{thm:OED_LB}
	There is no polynomial time randomized algorithm that can use $\OED$ queries alone to construct a $(1\pm \eps)$-approximate sparsifier of an underlying hypergraph $H$ with probability better than $o(1)$.
\end{theorem}

One may wonder if the oracle $\OED$ can be replaced with another access oracle that is used in sublinear algorithms for standard graphs, namely, ability to access the $i_{th}$ neighbor of a vertex $v$ for any integer $i$ that is at most the degree of $v$. It is easy to see that this is essentially same as the ability to access a random edge incident on a vertex $v$. 
We can generalize this idea to the setting of hypergraphs as follows. A neighbor query oracle in a hypergraph takes as input a set $S \subseteq V$, and returns a random edge that contains all vertices in $S$ if there is such an edge, and returns ${\rm NIL}$ if there is no edge. We say that a neighbor query is a {\em single vertex neighbor query} if $\card{S}=1$, and it is a {\em vertex pair neighbor query} if $\card{S}=2$. We denote the oracles that answer a single vertex neighbor query and a vertex pair neighbor query as $\ONONE$ and $\ONTWO$ respectively. We next show that the oracle $\OED$ can be replaced with the oracle $\ONTWO$, to obtain an alternate $\poly(n)$ time implementation of the result in Theorem~
\ref{thm:main}.

\begin{theorem} \label{thm:neighbor}
    Given an unweighted hypergraph $H=(V,E)$, suppose the algorithm can access the hypergraph using $\OVA$ and $\ONTWO$,
    then for any $0<\eps<1$, a $(1 \pm \eps)$-approximate sparsifier with $\tilde{O}(n/\eps^2)$ hyperedges can be constructed in $O(n^{10}/\eps^7)$ time in $n$, independent of the number of hyperedges.
\end{theorem}

In contrast to the result above, we show any algorithm that has access only to oracles $\OVA$ and $\ONONE$, requires exponentially many queries in the worst-case to construct a $\poly(n)$ size sparsifier.

\begin{theorem}
\label{thm:OVA_LB}
	There is no polynomial time randomized algorithm that can use $\OVA$ and $\ONONE$ queries alone to construct a $(1\pm \eps)$-approximate sparsifier of an underlying hypergraph $H$ with probability better than $o(1)$.
\end{theorem}

\noindent\textbf{Hypergraph Spectral Sparsification:} We also consider the problem of hypergraph spectral sparsification, a notion that strengthens cut sparsification. A $(1 \pm \eps)$-approximate {\em spectral sparsifier} of a graph $G(V,E)$ is a weighted graph $G'(V,E')$ such that for every vector $x \in \mathbb{R}^{n}$, we have 
$$  ~~~ |x^T L_{G'} ~x - x^T L_G ~x | \leq \eps (x^T L_G ~x) ,$$
where $L_G$ and $L_{G'}$  denote the Laplacian matrices of $G$ and $G'$, respectively. To see that the notion of spectral sparsifier only strengthens the notion of a cut sparsifier, observe that the cut sparsification requirement for any cut $(S, \bar{S})$ is captured by the definition above when we choose $x$ to be the $0/1$-indicator vector of the set $S$.
Batson, Spielman, and Srivastava~\cite{BatsonSS12} gave a polynomial-time algorithm that for every graph $G$, gives a weighted graph $G'$ with $O(n/\eps^2)$ edges such that $G'$ is a $(1 \pm \eps)$-approximate spectral sparsifier of $G$. Subsequently, Lee and Sun~\cite{LeeS17} gave an $O(m/\eps^{O(1)})$ time algorithm to construct a spectral graph sparsifier with $O(n/\eps^2)$ edges.

The notion of spectral sparsification can be extended to hypergraphs~\cite{Louis15,Yoshida19} as follows. The \emph{Laplacian} $L_H$ of a hypergraph $H$ is a function $\mathbb{R}^n \to \mathbb{R}^n$, such that for any $n$-dimensional vector $x$, we have
$$
x^TL_H(x) = \sum_{e \in E} w(e) \max_{u,v \in e}(x(u)-x(v))^2.
$$

Given a weighted hypergraph $H$ with $n$ vertices, a $(1 \pm \eps)$-spectral sparsifier $H'$ is a subgraph of $H$ such that for any $n$-dimensional vector $x$, we have

$$
(1-\eps)x^T L_{H'}(x) \le x^T L_H(x) \le (1+\eps)x^T L_{H'}(x).
$$

Soma and Yoshida~\cite{SomaY19} give a polynomial-time algorithm that outputs a weighted spectral sparsifier with $\tilde{O}(n^3)$ hyperedges. The algorithm of~\cite{SomaY19} is also based on sampling edges based on a suitable notion of importance, and in their work, the importance of a hyperedge $e$ is measured by $\min_{u,v \in e} \card{E(\{u,v\})}$ where $E(\{u,v\})$ is the set of edges that contains both $u$ and $v$.
If we assume access to the underlying hypergraph using $\ONTWO$ queries, then we can sample a random hyperedge in $E(\{u,v\})$ for any pair of vertices $u$ and $v$, which makes it in turn straightforward to simulate the algorithm of~\cite{SomaY19}.

\begin{theorem} \label{thm:spectral-neighbor}
    Given an unweighted hypergraph $H=(V,E)$, suppose the algorithm can access the hypergraph using $\OVA$ and $\ONTWO$ queries. Then for any $0<\eps<1$, a $(1 \pm \eps)$-spectral sparsifier with $\tilde{O}(n^3/\eps^2)$ hyperedges can be constructed in polynomial time in $n$, independent of the number of hyperedges.
\end{theorem}

The more interesting case is when we can only access the hypergraph using $\OVA$ and $\OED$ queries. It is now provably impossible to get a handle on $\card{E(\{u,v\})}$ using only polynomially many queries, and thus there is no direct way to simulate the algorithm in~\cite{SomaY19}. Recently, Bansal, Svensson, and Trevisan~\cite{BansalST19} designed another hypergraph spectral sparsification algorithm that in polynomial-time algorithm creates a weighted spectral sparsifier with $\tilde{O}(nr^3)$ hyperedges; here $r$ denotes the maximum arity of any hyperedge. Unlike the algorithm of~\cite{SomaY19}, their measure of importance of a hyperedge is derived from an auxiliary standard graph created by converting every hyperedge into a clique. The importance of a hyperedge is then given by the maximum effective resistance among all the edges in the clique associated with that hyperedge. We show that we can simulate this process using only $\poly(n)$ many $\OVA$ and $\OED$ queries.

\begin{theorem} \label{thm:spectral-edge}
    Given an unweighted hypergraph $H=(V,E)$, suppose the algorithm can access the hypergraph using only $\OVA$ and $\OED$ queries. Then for any $0<\eps<1$, a $(1 \pm \eps)$-spectral sparsifier with $\tilde{O}(n^3/\eps^2)$ hyperedges can be constructed in polynomial time in $n$, independent of the number of hyperedges.
\end{theorem}

The idea is to associate the strength of a hyperedge with the resistance of the edges inside the clique associated with this hyperedge. We first show that in any standard graph the effective resistance of an edge $f$ is at most $n/k_f$ where $k_f$ denote the strength of the edge $f$. We then show that the strength of a hyperedge $e$ is at most the strength of any edge in the clique associated with the hyperedge $e$. With this pair of relationships, we are able to simulate the algorithm in~\cite{BansalST19} by the algorithm in Theorem~\ref{thm:main}, except that we will sample somewhat larger number of hyperedges to meet the sampling probability requirement for the sparsification algorithm in~\cite{BansalST19}.

Note that unlike hypergraph cut sparsification which is now well-understood~\cite{CKN20}, the problem of determining the optimal size of a hypergraph spectral sparsifier is still open. However, note that any further improvements on the size of hypergraph spectral sparsifiers can be used in a black-box manner with our sparsification algorithms -- we can simply apply the improved sparsification algorithm to the sparsifier generated by our algorithm.

Finally, we note that since any spectral sparsifier is also a cut sparsifier, the lower bounds in Theorem~\ref{thm:OED_LB} and Theorem~\ref{thm:OVA_LB} also hold for spectral sparsifiers.

\medskip
\noindent
{\bf Extension to weighted hypergraphs:} All our algorithmic results stated thus far are for sparsifying unweighted hypergraphs. We note that all algorithms can be extended in a straightforward manner to sparsifying weighted hypergraphs assuming natural weighted versions of the access oracles used in proving these results. Specifically, it suffices to assume that the oracle $\OVA$ returns the {\em weight} of a cut, and the oracles $\OED$ and $\ONTWO$ return an edge with probability {\em proportional to its weight}.

\medskip
\noindent
{\bf Organization:} We set up our notation and state some useful background results in Section~\ref{sec:prelims}. 
In Section~\ref{sec:OED}, we give a $\poly(n)$ time algorithm for creating a $\tilde{O}(n/\eps^2)$ size sparsifier using the cut value and cut edge sample oracles, proving Theorem~\ref{thm:main}. In Section~\ref{sec:ONTWO}, we give a $\poly(n)$ time algorithm for creating a $\tilde{O}(n/\eps^2)$ size sparsifier using the cut value and the vertex pair neighbor oracles, proving Theorem~\ref{thm:neighbor}. 
Then in Section~\ref{sec:spectral}, we present $\poly(n)$ time algorithms for hypergraph spectral sparsification proving Theorems~\ref{thm:spectral-neighbor} and~\ref{thm:spectral-edge}.
We show in Section~\ref{sec:lowerbound} that any weakening of the oracles
assumed in Theorems~\ref{thm:main} and~\ref{thm:neighbor} necessarily requires worst-case exponential time for creating a $\poly(n)$ size hypergraph cut sparsifier, proving Theorems~\ref{thm:OED_LB} and~\ref{thm:OVA_LB}. In Section~\ref{sec:weighted}, we briefly describe how the results of Theorems~\ref{thm:main} and~\ref{thm:neighbor} can be extended to weighted hypergraphs. Finally, we conclude in Section~\ref{sec:conclusions}
with some directions for future work.

%% file: prelim.tex
\section{Preliminaries} \label{sec:prelims}
\subsection{Notation}
Given an integer $n$ and a probability $p$, let $B(n,p)$ be the Bernoulli distribution with $n$ trials where each trial succeeds with probability $p$. Suppose a set has $n$ elements, and given a probability $p$. Then the following process will sample each element in the set with probability $p$: we first sample a number $N \sim B(n,p)$, then randomly sample $N$ elements in the set.

Given any weight function $w:S\rightarrow \mathbb{R}_{\geq 0}$, we extend it to also be a function on subsets of $S$ so that $w(S') = \sum_{e\in S'}w(e)$ for $S'\subseteq S$. 

Given a graph $H = (V,E)$ and a subset of vertices $V'\subseteq V$, we define $G[V']$ to be the weighted subgraph/subhypergraph of $G$ induced by the vertices in $V'$. We identify an edge/hyperedge $e$ with the set of vertices that are contained in $e$. A hyperedge $e$ has \em{rank} $k$ if $\card{e}=k$. The \em{rank} of a hypergarph is the maximum rank of its hyperedges.

Given a weighted hypergraph $H = (V,E,w)$ and a cut $C = (S,\bar{S})$, we say an edge $e$ crosses the cut if $e \cap S \neq \emptyset$ and $e \cap \bar{S} \neq \emptyset$. We denote by $\delta_H(S)$ the set of the edges crossing the cut $C$ in $H$. By definition, $\card{\delta_H(S)}$ is the number of edges crossing $C$ and $w(\delta_H(S))$ is the weight of $C$. For any $\eps > 0$, a {\em $(1 \pm \eps)$-approximate cut sparsifier} of $H$ is a hypergraph $H' = (V,E',w')$ with $E' \subseteq E$ such that

$$\forall S \subseteq V, ~~~\card{ w'(\delta_{H'}(S))- w(\delta_H(S)) } \leq \eps w(\delta_H(S)).$$

Given a weighted normal graph $G$ with $n$ vertices, the \emph{Laplacian} $L_G \in \mathbb{R}^{n \times n}$ is defined as follows: for any $u$, $L_G(u,u)$ is the total weight edges incident on $u$, and for any $u \neq v$, $L_G(u,v)$ is minus the weight of edges between $u$ and $v$. For any $n$-dimensional vector $x \in \mathbb{R}^{n}$, $x^T L_G x = \sum_{(u,v) \in e} w(e)(x(u)-x(v))^2$. A $(1\pm \eps)$-spectral sparsifier $G'$ is a subgraph of $G$ such that for any vector $x$, 
$$(1-\eps)x^T L_{G'} x \le x^T L_G x \le (1+\eps)x^T L_{G'} x.$$

The notion of spectral sparsification can be extended to hypergraphs as follows. The \emph{Laplacian} $L_H$ of a hypergraph $H$ is a function $\mathbb{R}^n \to \mathbb{R}^n$, such that for any $n$-dimensional vector $x$, we have
$$
x^TL_H(x) = \sum_{e \in E} w(e) \max_{u,v \in e}(x(u)-x(v))^2.
$$

Given a weighted hypergraph $H$ with $n$ vertices, a $(1 \pm \eps)$-spectral sparsifier $H'$ is a subgraph of $H$ such that for any $n$-dimensional vector $x$, we have

$$
(1-\eps)x^T L_{H'}(x) \le x^T L_H(x) \le (1+\eps)x^T L_{H'}(x).
$$

\subsection{Hypergraph Cut Sparsification}
In this section, we review some important concepts and results of hypergraph cut sparsification. 

Given a weighted hypergraph $H=(V,E,w)$, a \emph{$k$-strong component} of $H$ is a maximal induced subgraph of $G$ that has minimum cut at least $k$. For any edge $e$, the \emph{strength} of $e$, denoted by $k_e$, in $H$ is the maximum value of $k$ such that $e$ is \emph{fully} contained in a $k$-strong component of $H$. Alternatively, the strength of an edge $e\in E$ is the largest minimum cut size among all induced subgraphs $H[X]$ that contain $e$, where $X$ ranges over all subsets of $V$. The sum of $w_e/k_e$ is at most $n-1$.

\begin{lemma} [\cite{kogan2015sketching}] \label{lem:n-1}
    Given a weighted hypergraph $H=(V,E,w)$, we have $\sum_{e \in E} w_e/k_e \le n-1$. 
\end{lemma}

Benczúr and Karger~\cite{benczur1996approximating, BenczurK15} showed that when we are dealing with a normal graph where each edge contains exactly two vertices, if we sample each edge $e$ independently with probability $p_e=O(\log n/\eps^2k_e)$, and give it weight $1/p_e$ if sampled, then the resulting graph will be a $(1 \pm \eps)$-approximate cut sparsifier of $G$ with high probability.

Kogan and Krauthgamer~\cite{kogan2015sketching} generalized this approach to hypergraphs. They showed that given a hypergraph $H=(V,E)$ with rank $r$, for each hyperedge $e$, if we sample $e$ with probability $p_e=O((\log n+r)/\eps^2 k_e)$, and has weight $1/p_e$ if get sampled, the resulting graph will be a cut sparsifier of $H$ with high probability.

\begin{theorem}[\cite{kogan2015sketching}] \label{thm:kk}
    Let $H$ be a hypergraph with rank $r$, and let $\eps>0$ be an error parameter. Consider the hypergraph $H'$ obtained by sampling each hyperedge $e$ in $H$ independently with probability $p_e=\min \{ 1, \frac{3((d+2)\log n +r)}{k_e \eps^2}\}$, giving it weight $1/p_e$ if included. Then with probability at least $1-O(n^{-d})$
\begin{enumerate}
    \item the hypergraph $H'$ has $O(\frac{n}{\eps^2}(r+\log n))$ edges, and
    \item $H'$ is a $(1 \pm \eps)$-approximate cut sparsifier of $H$.
\end{enumerate}
\end{theorem}

In fact, if for each edge $e$, the sampling proability $p_e$ is at least $\frac{3((d+2)\log n +r)}{k_e \eps^2}$, then the resulting graph is still a $(1 \pm \eps)$-approximate cut sparsifier. 
This is because in the proof of Theorem~\ref{thm:kk}, the authors showed that each cut has a very small probability that the cut size in $H'$ is not within a factor of $(1 \pm \eps)$ of the cut size in $H$, and the probability that there is no such cut is also very small by taking a union bound over all possible cuts. To bound the probability that a cut has a similar size in $H'$ and $H$, the authors use the Chernoff bound. However, we can also use the following concentration bound to prove the same result.

\begin{lemma}[Theorem $2.2$ in \cite{fung2019general}]\label{chernoff}
    Let $\{x_1,\ldots, x_k\}$ be a set of random variables, such that for $1 \le i \le k$, each $x_i$ independently takes value $1/p_i$ with probability $p_i$ and $0$ otherwise, for some $p_i \in [0,1]$. Then for all $N\geq k$ and $\eps\in (0,1]$,
    $$
    \prob{\card{\sum_{i\in [k]}x_i - k}\geq \eps N} \leq 2e^{-0.38\eps^2\cdot \min_i p_i\cdot N}
    $$
\end{lemma}

So if we replace the probability of sampling an edge $e$ with $q_e \ge p_e$, the concentration bound in Lemma~\ref{chernoff} still holds. In other words, if we sample according to $q_e$ then the probability that each cut in the sampled graph $H'$ has size close to the cut size in $H$ is at least as large as the probability when edges are sampled according to $p_e$.

\begin{lemma}\label{lem:kk}
Let $H$ be a hypergraph with rank $r$, and let $\eps>0$ be an error parameter. Consider the hypergraph $H'$ obtained by sampling each hyperedge $e$ in $H$ independently with probability $p_e \ge \min\{1, \frac{3((d+2)\log n +r)}{k_e \eps^2}\}$, giving it weight $1/p_e$ if included. Then with probability at least $1-O(n^{-d})$, $H'$ is a $(1 \pm \eps)$-approximate cut sparsifier of $H$.
\end{lemma}

Recently,~\cite{CKN20} showed that for every $n$-vertex hypergraph, there is a $(1\pm \eps)$-approximate cut sparsifier with $\tilde{O}(n)$ edges. Moreover, this sparsifier can be constructed in polynomial time in the number of vertices and the number of hyperedges.

\begin{theorem} [\cite{CKN20}] \label{thm:linear}
    Given a weighted hypergraph $H$, for any $0<\eps<1$, there exists an randomized algorithm that constructs a $(1 \pm \eps)$-approximate cut sparsifier of $H$ with $O(\frac{n\log n}{\eps^2})$ hyperedges in $O(mn+n^{10}/\eps^7)$ time with high probability.
\end{theorem}

\subsection{Hypergraph Spectral Sparsification}
\label{subsec:HSS}

In this section, we review some important concepts and results on spectral sparsification in both normal graphs and hypergraphs.

Given a weighted normal graph $G$, the \emph{effective resistance} $r_e$ of an edge $e=(u,v)$ is defined to be the electrical effective resistance between $u$ and $v$ if we view $G$ as a electrical network on $n$ nodes in which each edge $e$ corresponds to a resistor with conductance $w(e)$.

Spielman and Srivastava~\cite{SpielmanS11} showed that given an unweighted graph $G$, if we sample each edge $e$ independently with probability $p_e=O(r_e\log n/ \eps^2)$ and give it a weight of $w_e/p_e$ if sampled, then the resulting graph is a $(1 \pm \eps)$-spectral sparsifier.

In the case of hypergraphs, Soma and Yoshida~\cite{SomaY19} showed that if we sample each hyperedge $e$ with probability proportional to $n \log n/(\eps^2 \min_{u,v \in e} \card{E(\{u,v\})})$ where $E(\{u,v\})$ is the set of hyperedges that contains both $u$ and $v$, then the resulting graph is a $(1 \pm \eps)$-spectral sparsifier.

\begin{theorem} [\cite{SomaY19}] \label{thm:spectral-S19}
    Given an unweighted hypergraph $H$, if we sample each edge $e$ with probability $p_e=\min\{1,\frac{C n \log n}{\eps^2 \min_{u,v \in e} \card{E(\{u,v\})}}\}$ where $C$ is a universal constant, and give weight $1/p_e$ if sampled, then the resulting hypergraph $H'$ is a $(1 \pm \eps)$-spectral sparsifier of $H$ with high probability. The sparsifier has $\tilde{O}(n^3/\eps^2)$ hyperedges. The running time of this algorithm is $\tilde{O}(mn^2 + m + n^3/\eps^2)$.
\end{theorem}

Bansal et al.~\cite{BansalST19} take a different approach to hypergraph spectral sparsification. For any hypergraph $H$, define the auxiliary graph $G_H$ as a normal graph that is obtained by, for each hyperedge $e$, transforming $e$ into a clique $F_e$ over the vertices in $e$. For any hyperedge $e$, we now define $r_e = \max_{f \in F_e} r_f$. Bansal et al. showed that if we sample each hyperedge $e$ with probability proportional to $r^4 r_e \log n/\eps^2$ where $r$ is the maximum size of any hyperedge in $H$, then the resulting graph is a $(1 \pm \eps)$-spectral sparsifier.

\begin{theorem} [\cite{BansalST19}] \label{thm:spectral-B19}
    Given an unweighted hypergraph $H$, if we sample each edge $e$ with probability $p_e=\min\{1,\frac{C r^4 r_e \log n}{\eps^2}\}$ where $C$ is a universal constant, and assign it weight $1/p_e$ if sampled, then the resulting hypergraph $H'$ is a $(1 \pm \eps)$-spectral sparsifier of $H$ with high probability. The sparsifier has $\tilde{O}(nr^3/\eps^2)$ hyperedges.
\end{theorem}

Similar to the case of Theorem~\ref{thm:kk}, Theorem~\ref{thm:spectral-S19} and Theorem~\ref{thm:spectral-B19} both work when we sample each edge with a probability $q_e \ge p_e$ instead of $p_e$ and give weight $1/q_e$ if sampled.

\subsection{Minimizing a Submodular Function Using Value Queries}
A set function $f:2^{\Omega} \to \mathbb{R}$ is a {\em submodular function} if for every $S,T \subseteq \Omega$, $f(S)+f(T) \ge f(S \cup T) + f(S \cap T)$. Given a graph/hypergraph, for any vertex set $S$, the cut function $f(S)$, defined as the weight of edges crossing cut $(S,\bar{S})$ is easily shown to be submodular. There is an algorithm that for any submodular function $f$ finds a set $S$ that minimizes the value of function $f$ using $\tilde{O}(n^3)$ value queries and in $\tilde{O}(n^4)$ time.

\begin{theorem} [\cite{LeeSW15}] \label{thm:submodular}
    There is an algorithm for submodular function minimization with $O(n^3 \log^2 n)$ value queries and $O(n^4 \log^{O(1)} n)$ time where $n$ is the size of the ground set.
\end{theorem}

%% file: algorithm.tex
\section{Sublinear Time Cut Sparsification with Cut Size and Cut Edge Sampling Queries}
\label{sec:OED}

We now present an algorithm that, given access to a hypergraph $H$ through cut size queries (oracle $\OVA$) and queries to sample a random edge crossing a cut (oracle $\OED$), outputs a $(1 \pm \eps)$-approximate sparsifier with $\tilde{O}(n/\eps^2)$ hyperedges in $\poly(n)$ time. At a high-level, our algorithm will first create a $\poly(n)$ size sparsifier $H_1$ by indirectly implementing the algorithm underlying Theorem~\ref{thm:kk}. We then use the algorithm in Theorem~\ref{thm:linear} to construct a sparsifier $H_2$ of $H_1$ which has $\tilde{O}(n/\eps^2)$ hyperedges. By the definition of cut sparsifier, $H_2$ is also a cut sparsifier of $H$. We can thus focus on the construction of the sparsifier $H_1$.

The primary challenge in simulating the algorithm of Theorem~\ref{thm:kk} is to sample edges according to their strength with a small number of queries. Consider the following recursive algorithm. We start with the graph $H$, and then at each step, we find the minimum cut of the connected graph, and sample $\Theta((r+\log n)/\eps^2)$ edges from the cut. We then recursively execute this algorithm on each side of the cut. Algorithm~\ref{alg:strength} gives an implementation of this idea.

\begin{algorithm}[htp] 
    \SetAlgoLined
    \SetKwInOut{Input}{Input}\SetKwInOut{Output}{Output}
%\noindent
    \Input{A subset of vertices $V'\subseteq V$.}

    Let $(S,\bar{S})$ be a minimum cut of the induced graph $G[V']$\;
    Let $c$ be the number of edges crossing $(S,\bar{S})$ in $G[V']$\;
    Sample an integer $N\sim B(c, \frac{10(\log n + r)}{\epsilon^2 c})$\;
    Sample $N$ edges from $\delta_{G[V']}(S)$ uniformly at random, and assign each of them a weight of $\frac{\epsilon^2 c}{10(\log n + r)}$\;
    Delete all edges in $\delta_{G[V']}(S)$ and recurse on each of the newly created connected components\;
    \caption{Sampling edges with probability proportional to their strength} \label{alg:strength}
\end{algorithm}
\vspace{\baselineskip}

It is easy to see that this algorithm samples each edge independently, and that the sampling probability is at least that of Kogan-Krauthgamer in Theorem~\ref{thm:kk}. The challenge is that unlike cut queries in the normal graph, it is hard to compute the cut size in an induced subgraph of a hypergraph using only cut queries on the original graph, which is crucial as the algorithm proceeds recursively. 

We first note that this task is straightforward to do in graphs where each edge has exactly two vertices. For any two disjoint subsets of vertices $S,T$, the number of edges in $S \times T$ is $\frac{1}{2}(\card{\delta(S)}+\card{\delta(T)}-\card{\delta(S \cup T)})$. 

However, this is far from true in the hypergraph setting.
The problem is that there may be some hyperedges that intersect with each of $S$, $T$, and $V \setminus (S \cup T)$. These edges are inside all of $\delta(S)$, $\delta(T)$ and $\delta(S \cup T)$. We have \begin{align*}
    & \card{\delta(S)}+\card{\delta(T)}-\card{\delta(S \cup T)} \\
    = & 2\card{\{e | e \cap S \neq \emptyset, e \cap T \neq \emptyset, e \cap (\overline{S \cup T}) = \emptyset \}} + \card{\{e | e \cap S \neq \emptyset, e \cap T \neq \emptyset, e \cap (\overline{S \cup T}) \neq \emptyset \}} \\
    = & \card{\{e | e \cap S \neq \emptyset, e \cap T \neq \emptyset\}} + \card{\{e | e \cap S \neq \emptyset, e \cap T \neq \emptyset, e \cap (\overline{S \cup T}) = \emptyset \}}
\end{align*}

\begin{example} \label{exm:1}
Consider a hypergraph $H$ that consists of three equal size sets of vertices $A,B,C$, such that each hyperedge has a non-empty intersection with each of $A$, $B$, and $C$. Then there are no hyperedges in $H[A \cup B]$. But the quantity $\frac{1}{2}(\card{\delta(A)}+\card{\delta(B)}-\card{\delta(A \cup B)})$ is half the total number of hyperedges which could be exponentially large in $n$. 
\end{example}

\begin{example} 
Consider the following pair of hypergraphs on $4$ vertices, say $\{ v_1, v_2, v_3, v_4\}$: the graph $H_1$ is a (rank 2) clique on $4$ vertices while the graph $H_2$ contains every possible edge of size $3$ on these $4$ vertices. It is easy to verify that the answer to every cut query is the same on the graphs $H_1$ and $H_2$. Now consider the subgraph of these graphs induced by the vertices $X = \{ v_1, v_2 \}$. In case of $H_1$, the minimum cut in the induced subgraph is $1$ while in $H_2$, the minimum cut in the graph induced by $X$ is $0$. We can amplify this gap to $0$ versus $\Omega(n)$ by taking $n/4$ copies of $H_1$ in one case, and $n/4$ copies of $H_2$ in the other case, and defining $X$ to be union of arbitrarily chosen pairs of vertices from each copy. This means that $\OVA$ queries can not be used to estimate cut size in induced subgraphs to any multiplicative factor or to better than a polynomial additive error.
\end{example}

To get around the challenge highlighted by examples above, we next introduce notions of \textit{pseudo cut size} over a subset of vertices and \textit{pseudo strength} of hyperedges, such that the pseudo cut sizes are easy to compute by cut queries and pseudo strength of any hyperedge is at most a factor $n$ larger than the strength of the hyperedge. We develop these ideas in detail in the next subsection.

\subsection{Pseudo Cuts and Pseudo Strengths}
Given a set of vertices $X$, we define $\Delta_X(S)$, the \textit{pseudo cut} size of a set $S \subset X$ as $\frac{1}{2}(\card{\delta(S)}+\card{\delta(X \setminus S)}-\card{\delta(X)})$, and define the \textit{pseudo min cut} over $X$ as a cut $(S,X \setminus S)$ that minimizes $\Delta_X(S)$. Note that $\Delta_X(S)$ is at most the number of edges that intersect both $S$ and $X \setminus S$. The following lemma shows that $\Delta_X(S)$ is a submodular function, so we can compute the pseudo min cut over any vertex set in $\poly(n)$ time by Theorem~\ref{thm:submodular}.

\begin{lemma} \label{lem:pes-sub}
    For any vertex set $X \subseteq V$, $\Delta_X(S)$ is a submodular function.
\end{lemma}

\begin{proof}
    Let $f_1(S)$ be the number of edges that intersect both $S$ and $X \setminus S$, and let $f_2(S)$ be the number of edges that intersect both $S$ and $X \setminus S$ but are fully contained in $X$. By definition, we have $\Delta_X(S) = \frac{1}{2}(f_1(S)+f_2(S))$, so to prove that $\Delta_X(S)$ is a submodular function, it is sufficient to prove that both $f_1$ and $f_2$ are submodular.

    Since $f_2$ is the cut function in the induced graph $H[X]$, it is submodular. In fact, $f_1$ is also the cut function of the hypergraph whose vertex set is $X$ and edge set is $\{e \cap X | e \in E\}$. So $f_1$ is also a submodular function.
\end{proof}

For any edge $e$, we define the \textit{pseudo strength} $k'_e$ as the largest pseudo min-cut size among all sets $X$ that contain $e$, where $X$ ranges over all subsets of $V$. It is easy to see that for any edge $e$, $k'_e$ is at least $k_e$ since for any set of vertices $X$, the minimum cut size of $H[X]$ is at most the pseudo min-cut size of set $X$. More interestingly, although Example~\ref{exm:1} showed that the pseudo min-cut size of a set $X$ may be arbitrarily larger than the minimum cut size of $H[X]$, the lemma below shows that the pseudo strength of an edge is at most a factor $n$ larger than its strength.

\begin{lemma} \label{lem:pes-bound}
    For any edge $e$, $k'_e \le n k_e$.
\end{lemma}

\begin{proof}
    Let $X$ be any set of vertices that contains the edge $e$ and has pseudo min-cut size $k'_e$ in $H$. To prove the lemma, it is sufficient to prove that the pseudo min-cut size of $X$ in $H$ is at most $n k_e$.

    Let $Y=V$ and $E_c=\emptyset$. Consider the following iterative process: we find the minimum cut $(S,Y \setminus S)$ in $H[Y]$. If either $S$ or $Y \setminus S$ fully contains the set $X$, we add all edges crossing the cut into $E_c$, and set $Y$ to be $S$ or $Y \setminus S$ (whichever fully contains $X$), and repeat. Otherwise, we stop the process.

    After the process terminates, suppose $(S,Y \setminus S)$ is the minimum cut in $H[Y]$. Since the process terminated, $(S, Y\setminus S)$ must partition $X$. Let $S'=S \cap X$, and consider the pseudo cut $(S',X \setminus S')$. We prove that the number of edges in $H$ that intersect both $S'$ and $X \setminus S'$ (which is an upper bound on $\Delta_X(S')$) is at most $n k_e$. 

    First, note that no edge $e'$ such that $e'\not\subseteq Y$ and $e'\notin E_c$ can intersect with the set $Y$; hence any such edge $e'$ also does not intersect with $S'$ or $X \setminus S'$. Therefore every edge that intersects with both $S'$ and $X\setminus S'$ either belongs to $E_c$ or is completely contained in $Y$. During the iterative process, the set $Y$ always fully contains $e$, so by the definition of strength, the minimum cut size of $H[Y]$ is at most $k_e$. This implies during each step, at most $k_e$ edges are added into $E_c$. On the other hand, the process repeats at most $n-2$ times, since each time the size of $Y$ is reduced by at least 1. So $\card{E_c} \le (n-2)k_e$. Finally, any edge that is fully contained in $Y$ and intersects with both $S'$ and $X \setminus S'$ crosses the cut $(S,Y \setminus S)$ in $H[Y]$, and the number of such edges is at most the minimum cut size of $H[Y]$, which is at most $k_e$. So in total, there are at most $\card{E_c} + k_e\le (n-1)k_e$ edges that intersect both $S'$ and $X \setminus S'$.
\end{proof}

\subsection{Sampling the Edges}

We are now ready to present an algorithm that uses the cut size queries and cut edge sample queries to sample each edge with probability inversely proportional to its strength. Specifically, we will ensure that each edge $e$ gets sampled with probability at least $n^2/k'_e$ which is at least $n/k_e$ by Lemma~\ref{lem:pes-bound}. The algorithm is similar to Algorithm~\ref{alg:strength}, but uses pseudo cuts and pseudo strengths instead. To sample the edges, we call Algorithm~\ref{alg:pesudo} on set $V$.

\begin{algorithm}[htp] 
    \SetAlgoLined
    \SetKwInOut{Input}{Input}\SetKwInOut{Output}{Output}

    \Input{A subset of vertices $V'\subseteq V$}

    Find the pseudo min-cut $(S, V' \setminus S)$ within the set $V'$\;
    Let $c$ be $\card{\delta(S)}$, the cut size of $(S,V\setminus S)$ \;
    Sample an integer $N\sim B(c, \min\{1,\frac{10n^3}{\epsilon^2 c}\})$\;
    Keep sampling edges in cut $(S,\bar{S})$ until we get $N$ different hyperedges. \;
    Recurse on both $S$ and $V' \setminus S$ \;
    \caption{Sampling edges with probability proportional to their pseudo strength} \label{alg:pesudo}
\end{algorithm}

\vspace{\baselineskip}

We now prove that each edge gets sampled with probability at least as large as the sampling probability in Theorem~\ref{thm:kk}. Fix an edge $e$, let $S_1$ be the last input set that fully contains $e$. For any $i \ge 1$, if $S_i$ is not $V$, we define $S_{i+1}$ to be the input set in the recursion that generates a recursive call of the algorithm on the set $S_i$. In other word, $(S_i,S_{i+1} \setminus S_i)$ is the pseudo min cut within set $S_{i+1}$. Let $(S_0,S_1 \setminus S_0)$ be the pseudo min cut within $S_1$, by definition, $e \cap S_0 \neq \emptyset$ and $e \cap S_1 \setminus S_0 \neq \emptyset$. When the algorithm works on set $S_1$, $e$ gets sampled with probability $\min\{1,\frac{10 n^3}{\eps^2\card{\delta(S_0)}}\}$. If $e$ gets sampled with probability $1$, then it is clearly as large as the probability in Theorem~\ref{thm:kk}. Otherwise we need to prove that $n^3/\card{\delta(S_0)} = \Omega( (\log n + r)/k_e)$. Since $n = \Omega(\log n+r)$, by Lemma~\ref{lem:pes-bound}, it is sufficient to prove that $\card{\delta(S_0)} \le n k'_e$. 

\begin{lemma} \label{lem:sample}
    $\card{\delta(S_0)} \le n k'_e$.
\end{lemma}

\begin{proof}
    We partition the edges crossing the cut $(S_0,\overline{S_0})$ into sets $E_1,E_2,\dots$ such that for any $i \ge 0$, $E_i$ is the set of edges that are fully contained in $S_{i+1}$ but not in $S_i$. Note that $\card{\delta(S_0)} = \sum_i \card{E_i}$. Since the algorithm has at most $n$ levels of recursion, to prove the lemma, it is sufficient to prove $\card{E_i} \le k'_e$ for all $i \ge 0$.

    For any edge $e' \in E_i$, $e' \cap S_i \neq \emptyset$ since $e'$ crosses the cut $(S_0,\overline{S_0})$ and $S_0 \subseteq S_i$. We also have $e' \cap S_{i+1} \setminus S_i \neq \emptyset$ and $e' \cap \overline{S_{i+1}} = \emptyset$ since $e'$ is fully contained in $S_{i+1}$ but not $S_i$. So $\card{E_i} \le \Delta_{S_{i+1}}(S_i)$. On the other hand, by definition of pseudo strength, $k'_e \ge \Delta_{S_{i+1}}(S_i)$ since $e$ is fully contained in $S_{i+1}$. Therefore, $\card{E_i} \le k'_e$.
\end{proof}

By Lemma~\ref{lem:sample}, we proved that each edge $e$ is sampled with probability at least the required probability in Theorem~\ref{thm:kk}. Next, we need to assign weights to each sampled edge. 

We do this after we finish sampling. For each edge $e$ that gets sampled, we need to know the probability that it gets sampled. Since we sample edges from each cut independently, we only need to know the probability that $e$ gets sampled during each recursive call, and that probability depends only on the size of the cut and whether $e$ crosses the cut. So we can compute the probability that $e$ gets sampled during Algorithm~\ref{alg:pesudo}.

To complete the proof of Theorem~\ref{thm:main}, we need to show that the running time of the whole process is polynomial in $n$.

\begin{proof}[Proof of Theorem~\ref{thm:main}]
    During each call to Algorithm~\ref{alg:pesudo}, we need $\tilde{O}(n^3)$ queries to cut size query oracle and $\tilde{O}(n^4)$ time to figure out the pseudo min-cut within the set $V'$ by Theorem~\ref{thm:submodular} and Lemma~\ref{lem:pes-sub}. At line 4, we call cut edge sample query $10n^3/\eps^2$ times in expectation. Total number of recursive calls to Algorithm~\ref{alg:pesudo} is $O(n)$, since each time, the input set gets partitioned into two sets, and there are $n$ sets in the end. Thus the running time of Algorithm~\ref{alg:pesudo} is $\tilde{O}(n^5 + n^4/\eps^2)$.

    We sample the edges in $O(n)$ cuts, so when assigning the weights, we only need to query the size of these $O(n)$ cuts and calculate the probability of each sampled edge, which can also be done in $O(n)$ time for each edge. So the running time of assigning the weights is $\tilde{O}(n^5/\eps^2)$.

    After sampling the edges and assigning weights, we get a $(1 \pm \eps)$-approximate cut sparsifier $H_1$ of $H$ with polynomial size in $n$. Then we run the algorithm in Theorem~\ref{thm:linear} to find a $(1+\eps)$-approximate cut sparsifier $H_2$ of $H_1$ with $\tilde{O}(n/\eps^2)$ number of edges in polynomial time in $n$. By definition of cut sparsifier, $H_2$ is a $(1 \pm \eps)^2$-approximate cut sparsifier of $H$. Since $H_1$ contains $\tilde{O}(n^4/\eps^2)$ edges, by Theorem~\ref{thm:linear}, the running time is $O(n^{10}/\eps^7)$.

    So the total running time is $O(n^{10}/\eps^7)$.
\end{proof}

%% file: neighbor.tex
\section{Sublinear Time Cut Sparsification with Cut Size and Pair Neighbor Queries}
\label{sec:ONTWO}

In this section, we show that cut edge queries can be simulated by a $\poly(n)$ number of cut size queries (oracle $\OVA$) and pair neighbor queries (oracle $\ONTWO$), establishing that cut size query oracle and pair neighbor query oracle are also sufficient to compute a $(1 \pm \eps)$-approximate cut sparsifier in $\poly(n)$ time. 

Given a pair of vertices $u$ and $v$, let $E(\{u,v\})$ be the set of edges that contain both $u$ and $v$. We first show how to approximate $|E(\{u,v\})|$ to within a factor of $(1\pm \eps)$ with probability $1-\xi$, for some small $\xi$. Note that we can compute $2 \Delta_{\{u,v\}}(\{u\}) = |E(\{u,v\})| + |E \cap \{u,v\}|$, where $|E\cap \{u,v\}|$ is the number of copies of the edge $\{u,v\}$. We now describe an algorithm to approximate $|E(\{u,v\})|$:

\begin{algorithm}[H] \label{alg:pair}
\SetAlgoLined
	\SetKwInOut{Input}{Input}\SetKwInOut{Output}{Output}

	\Input{A pair of vertices $u,v\in V$, $\eps$, $\xi$.}
    \Output{An approximation $\hat{E}(\{u,v\})$ of $|E(\{u,v\})|$.}
	
	Define $k = 12\log(2/\xi)/(\eps^2)$.

    Call the oracle $\ONTWO$ $k$ times on $(u,v)$, and let $\hat{\alpha}$ be the fraction of returned edges that were $\{u,v\}$.

    Return $\hat{E}(\{u,v\}):=2\Delta_{\{u,v\}}(\{u\})\cdot \frac{1}{1+\hat{\alpha}}$.
 \caption{Approximating $|E(\{u,v\})|$}
\end{algorithm}

Note that this algorithm makes $k = O(\frac{\log(1/\xi)}{\eps^2})$ queries.

\begin{lemma}
    With probability at least $1-\xi$, $\hat{E}(\{u,v\})$ is an approximation of $|E(\{u,v\})|$ to within a factor of $(1 \pm \eps)$.
\end{lemma}
\begin{proof}

    Let $\alpha := \frac{|E\cap \{u,v\}|}{|E(\{u,v\})|}$ be the fraction of hyperedges that are $\{u,v\}$. The algorithm runs a Monte Carlo simulation to approximate $\alpha$ by the ratio $\hat{\alpha}$. In order to prove concentration of $\hat{\alpha}$ around $\alpha$, let $k'$ be the total number of $\{u,v\}$ edges returned, and observe that $k'$ is the sum of $k$ independent Bernoulli random variables each having probability equal to $\alpha$. By Chernoff bound, $\Pr[|k' - \alpha k| > \eps k/2]\leq 2\exp(-\alpha k \eps^2/12\alpha) \leq 2\exp(-k \eps^2/12) = 2\exp(\log(\xi/2)) = \xi$. Therefore with probability at least $1-\xi$, $|\hat{\alpha} - \alpha| \leq \eps/ 2$. This implies that $\frac{1}{1+\hat{\alpha}}\in \frac{1}{1+\alpha \pm \eps/2} \subseteq \frac{1}{(1 \pm \eps/2)(1+\alpha)}$. Finally, we use that $\frac{1}{(1\mp \eps/2)}\subseteq (1\pm \eps)$ to conclude that $\frac{1}{1+\hat{\alpha}}\in \frac{1\pm \eps}{1+\alpha}$, so 
	\[\frac{2\Delta_{\{u,v\}}(\{u\})}{1+\hat{\alpha}}\in (1\pm \eps)\frac{2\Delta_{\{u,v\}}(\{u\})}{1+\alpha} = (1\pm \eps)|E(\{u,v\})|.\]

\end{proof}

We now describe an algorithm to sample a random edge from $\delta(S)$, simulating a response to $\OED$. We first approximate the size of $E(\{u,v\})$ for each pair of vertices $u\in S$ and $v \in \bar{S}$. Then we sample a pair of $u,v$ with probability proportional to $\card{E(\{u,v\})}$, sample an edge in $E(\{u,v\})$, and then decide whether we keep it or not with probability proportional to its size. If we decide not to pick the edge, we repeat the whole process again.

 \begin{algorithm}[H]\label{alg:sample}
\SetAlgoLined
	\SetKwInOut{Input}{Input}\SetKwInOut{Output}{Output}

	\Input{A subset $S\subseteq V$}
	\Output{An edge $e\in \delta(S)$}
        For each pair of vertices $u,v$ such that $u \in S$ and $v \in \bar{S}$, call Algorithm~\ref{alg:pair} with $\xi = 1/n^{20}$, and let $\hat{E}(\{u,v\})$ be the output\;
        Sample a pair of vertices $(u,v)\in S\times \bar{S}$ with probability proportional to $\hat{E}(\{u,v\})$\;
		Use the oracle $\ONTWO$ to sample an edge $e$ in $E(\{u,v\})$\;
        With probability $\frac{1}{|e\cap S|\cdot |e\cap \bar{S}|}$, return	$e$. Otherwise go to Step 2.
	
 \caption{Sampling an edge in $\delta(S)$}
\end{algorithm}

\begin{lemma}
    With probability at least $1 - 1/n^{-10}$, Algorithm~\ref{alg:sample} samples each edge in $\delta(S)$ gets with probability $\frac{1 \pm \eps}{\card{\delta(S)}}$. The expected running time is $\tilde{O}(n^2/\eps^2)$.
\end{lemma}
\begin{proof}
    We first condition on the $|S|\cdot |\bar{S}|\leq n^2$ events that for each pair $u, v$ with  $u\in S$ and $v\in \bar{S}$, the estimate $\hat{E}(\{u,v\})$ was indeed in $(1\pm \eps)\cdot |E(\{u,v\})|$, which happens with probability at least $1 - n^2 \xi > 1- 1/n^{-10}$. Now fix an edge $e \in \delta(S)$. The probability that it was sampled at a particular iteration of Algorithm~\ref{alg:sample} is
    \[ \sum_{u\in S\cap e, v\in \bar{S}\cap e}\frac{\hat{E}(\{u,v\})}{\sum_{(u',v') \in S \times \bar{S}}\hat{E}(\{u',v'\})}\cdot \frac{1}{|E(\{u,v\})|}\cdot \frac{1}{|e\cap S|\cdot |e\cap \bar{S}|} \]
    \[ \in \frac{1}{\sum_{(u',v') \in S \times \bar{S}}\hat{E}(\{u',v'\})}\sum_{u\in S\cap e, v\in \bar{S}\cap e}\frac{(1\pm \eps)}{|e\cap S|\cdot |e\cap \bar{S}|} = \frac{(1\pm \eps)}{\sum_{(u',v') \in S \times \bar{S}}\hat{E}(\{u',v'\})}\]

That is, the probability of sampling each edge at any given iteration of Algorithm~\ref{alg:sample} is within $(1\pm \eps)$ of every other edge. Therefore the probability of sampling each edge is within a factor of $(1\pm \eps)$ of every other edge.

    Step $1$ calls Algorithm~\ref{alg:pair} $O(n^2)$ times, so the running time is $\tilde{O}(n^2/\eps^2)$. At step $4$, the probability that we keep the edge and finish the algorithm is at least $1/n^2$, so the expected number of iterations through step 2 to 4 is at most $n^2$. So the total running time on step 2 to 4 is $\tilde{O}(n^2)$ in expectation.
\end{proof}

\begin{proof}[Proof of Theorem~\ref{thm:neighbor}]
    We run Algorithm~\ref{alg:pesudo}, but each time it calls $\OED$, we instead run Algorithm~\ref{alg:sample} twice. With high probability, each time we simulate $\OED$ by Algorithm~\ref{alg:sample}, the probability of any edge in the cut being sampled is within a $(1\pm \eps)$ factor of the uniform distribution. Denote by $q'_e$ be the probability that an edge $e$ is sampled by this algorithm, and let $q_e$ be the probability that the edge $e$ is sampled in Algorithm~\ref{alg:pesudo}. We have $q'_e \in 2(1 \pm \eps) q_e$, which is larger than $p_e$ the probability of sampling an edge in Theorem~\ref{thm:kk}. Also we cannot directly compute $q'_e$, but we can approximate it to within a factor of $(1 \pm \eps)$, which only adds another $(1 \pm \eps)$ factor to the approximation achieved by the cut sparsifier.

    Since the number of calls to $\OED$ oracle in Algorithm~\ref{alg:pesudo} is $\tilde{O}(n^4/\eps^2)$. So we need $\tilde{O}(n^6/\eps^2) = o(n^{10}/\eps^7)$ queries to simulate these calls. So the running time of the algorithm is still $O(n^{10}/\eps^7)$.
\end{proof}

%% file: spectral.tex
\section{Sublinear Time Hypergraph Spectral Sparsification} 
\label{sec:spectral}

In this section, we consider the problem of hypergraph spectral sparsification in the access models considered in the previous sections. We will focus on unweighted hypergraphs here, and show that these results can be extended to the weighted case in Section~\ref{sec:weighted}. We first consider the setting when we can access the underlying hypergraph using $\OVA$ and $\ONTWO$ queries. It is easy to simulate the approach in Theorem~\ref{thm:spectral-S19} since $\ONTWO$ allows us to sample from $E(\{u,v\})$ for any $u$ and $v$. For any pair of vertices $u$ and $v$, we sample $Cn \log n /\eps^2$ edges in $E(\{u,v\})$. For any hyperedge $e$, $e$ gets sampled with probablity $q_e$ which is at least the required value $p_e$. Next we need to assign weights to the sampled edges. For any pair of vertices $u$ and $v$, we use Algorithm~\ref{alg:pair} to approximate $\card{E(\{u,v\})}$. Then for any $e$, we approximate $q_e$ by the approximation of $\card{E(\{u,v\})}$ for all pair $u$ and $v$ in $e$, and then assign the weight of $e$ as $1/q_e$.

\begin{proof} [Proof of Theorem~\ref{thm:spectral-neighbor}]
    For each pair of vertices $u$ and $v$, we sample $\tilde{O}(n)$ edges in $E(\{u,v\})$, and also use Algorithm~\ref{alg:pair} to approximate $\card{E(\{u,v\})}$. The total number of edges in the sparsifier and the total number of queries we perform is $\tilde{O}(n^3)$. For every sampled edge $e$, we need to calculate $q_e$, which costs at most $O(n^2)$ time as we need to combine the probabilities of sampling $e$ through any pair of vertices inside $e$. So the total time complexity of the algorithm is $\tilde{O}(n^5)$.
\end{proof}

We next consider the case when we are given access to hypergraph via $\OVA$ and $\OED$ queries only. We observe that we cannot simulate the algorithm in Theorem~\ref{thm:spectral-S19}. Consider the following hypergraph $H$: there is a pair of vertices $u$ and $v$, such that $H$ contains all possible hyperedges that do not contain both $u$ and $v$. $H$ also contains an edge $\{u,v\}$. Since $\{u,v\}$ is the only edge that contains both $u$ and $v$, we need to sample this edge with probability $1$ in the algorithm of Theorem~\ref{thm:spectral-S19}. However, any cut in the hypergraph has exponential (in $n$) size. So we need an exponential number of queries to sample the edge $\{u,v\}$.

Given the obstacle above, we will instead simulate Theorem~\ref{thm:spectral-B19}. Our approach is to show that the task of implementing Theorem~\ref{thm:spectral-B19} can be accomplished by our algorithm in Section~\ref{sec:OED}, except that we will sample $\poly(n)$ times more hyperedges. By doing so, we will guarantee that the probability that each hyperedge $e$ gets sampled is larger than the $p_e$ in Theorem~\ref{thm:spectral-B19}. We now develop the ideas needed to establish this coupling between Theorem~\ref{thm:spectral-B19} and our algorithm in Section~\ref{sec:OED}.

We first observe a relationship between the effective resistance and strength of an edge in a normal graph.

\begin{lemma} \label{lem:resistance}
    For any edge $f$ in a normal (weighted) graph $G$, we have $r_f \le \frac{n}{k_f}$.
\end{lemma}

\begin{proof}
    Suppose the edge $f = (u,v)$, and let $c$ denote the min-cut size between $u$ and $v$. Since for any vertex-induced subgraph of $G$ that contains both $u$ and $v$, the min-cut size is at most $c$, it follows that $k_f \le c$. So it is sufficient to prove that $r_f \le \frac{n}{c}$. 

    Since the min-cut size between $u$ and $v$ is $c$, the max-flow size between $u$ and $v$ is also $c$, which means we have a set of $c$ edge-disjoint paths from $u$ to $v$. As each path can have length at most $n$, this set of edge-disjoint paths can be interpreted as $c$ parallel resistors each with resistance at most $n$. By Rayleigh's monotonicity law, it then follows that $r_f \le \frac{n}{c}$.
\end{proof}

We now briefly review the approach underlying Theorem~\ref{thm:spectral-B19} (see Section~\ref{subsec:HSS}). Recall that for any hyperedge $e$ in a hypergraph $H$, $F_e$ is the clique associated with $e$ in the auxiliary graph $G_H$, and $r_e = \max_{f \in F_e} r_f$. Define $\kappa_e = \min_{f \in F_e} k_f$. Then by Lemma~\ref{lem:resistance}, we have $r_e \le \frac{n}{\kappa_e}$. The following lemma shows a relationship between $\kappa_e$ and $k_e$.

\begin{lemma} \label{lem:kappa}
    For any hyperedge $e$ in a hypergraph $H$, $\kappa_e \ge k_e$.
\end{lemma}

\begin{proof}
    For any subset of vertices $X$, consider the corresponding vertex-induced subgraphs of $H$ and $G_H$, which we denote by $H[X]$ and $G_H[X]$, respectively. For any cut defined by a partition of $X$, and for any hyperedge $e$ crossing this cut in $H[X]$, at least one edge in $F_e$ must cross this cut in $G_H[X]$. So the min-cut size of $H[X]$ is at most the min-cut size of $G_H[X]$. Let $X_e$ be the set of vertices such that $e \subseteq X_e$ and the min-cut size of $H[X_e]$ equals $k_e$. The min-cut size of $G_H[X_e]$ is at least $k_e$. For any edge $f \in F_e$, $f$ is contained in $X_e$ and so $k_f \ge k_e$, which means $\kappa_e = \min_{f \in F_e} k_f \ge k_e$.
\end{proof}

By Lemma~\ref{lem:resistance} and Lemma~\ref{lem:kappa}, we have $\frac{n}{k_e} \ge r_e$ for any hyperedge $e$. So by Theorem~\ref{thm:spectral-B19}, if we sample each hyperedge with probability $q_e$ which is at least $\min\{1,\frac{Cn^5\log n}{\eps^2 k_e}\}$, and assign it weight $1/q_e$ if sampled, then the resulting hypergraph is a $(1 \pm \eps)$-spectral sparsifier. So we can use the same process as described in Section~\ref{sec:OED}, except that we oversample hyperedges by a factor of $n^4$. The resulting hypergraph will then be a $(1 \pm \eps)$-spectral hypergraph sparsifier.

\begin{proof} [Proof of Theorem~\ref{thm:spectral-edge}]
    We run Algorithm~\ref{alg:pesudo}, except that in line 3, we set $N \sim B(c,\min \{1,\frac{Cn^7}{\eps^2 c}\})$ where $C$ is the constant in Theorem~\ref{thm:spectral-B19}. By the same argument as in Section~\ref{sec:OED}, we sample each edge with probability at least $\min \{1,\frac{Cn^5 \log n}{\eps^2 k_e}\}$, which means the resulting graph is a $(1 \pm \eps)$-spectral sparsifier if we assign the weight of each edge also using the same process in Section~\ref{sec:OED}. The number of hyperedges sampled is $\tilde{O}(n^8/\eps^2)$. We then run the algorithm in Theorem~\ref{thm:spectral-S19} on our sparsifier and get a $(1 \pm 2\eps)$-spectral sparsifier with $\tilde{O}(n^3/\eps^2)$ hyperedges. The time taken by running Algorithm~\ref{alg:pesudo} and assigning weights is $\tilde{O}(n^9/\eps)$ since there are $n^4$ times more hyperedges being sampled. The running time of the algorithm in Theorem~\ref{thm:spectral-S19} is $\tilde{O}(n^{10})$ since there are $\tilde{O}(n^8)$ hyperedges sampled by Algorithm~\ref{alg:pesudo}. So the overall running time is $\tilde{O}(n^{10})$.
\end{proof}

%% file: lowerbound.tex
\section{Lower Bounds} \label{sec:lowerbound}

In this section we show that any natural relaxation of the assumptions underlying Theorem~\ref{thm:main} and~\ref{thm:neighbor} rules out $\poly(n)$ time sparsification algorithms, proving Theorem~\ref{thm:OED_LB} and~\ref{thm:OVA_LB}.

\subsection{Queries $\OVA$ and $\ONONE$ Together are not Sufficient}

In this section, we prove that if any randomized algorithm can only access the underlying hypergraph via $\OVA$ and $\ONONE$, it is not possible to find with probability better than $o(1)$ a $(1 \pm \eps)$-approximate cut sparsifier with only   $\poly(n)$ queries, proving Theorem~\ref{thm:OVA_LB}. We start by showing a weaker result, as stated in the lemma below, which shows that the failure probability of a $\poly(n)$ time algorithm must be at least $1/2 - o(1)$, and then show how to amplify the failure probability to $1 - o(1)$.

\begin{lemma}
	There is no polynomial time algorithm that can use $\OVA$ and $\ONONE$ queries alone to construct a $(1\pm \eps)$-approximate sparsifier of an underlying hypergraph $H$ with probability at least $1/2+\xi$ for any constant $\xi > 0$. 
\end{lemma}

\begin{proof}
	Suppose the runtime of the algorithm is bounded by some polynomial $f(n)$. We will construct two graphs $H_1=(V\cup V', E_1)$ and $H_2 = (V\cup V', E_2)$ with $|V| = |V'| = n$ and the algorithm is shown with probability $1/2$ the graph $H_1$ and with probability $1/2$ the graph $H_2$. We will then show that (a) any algorithm that can only access the underlying graph using $\OVA$ and $\ONONE$ cannot distinguish between these two graphs with probability at least $1/2+\xi$ for any constant $\xi > 0$, and (b) there exists a non-empty cut such that $H_1$ and $H_2$ do not have any common edges crossing the cut. Together, these properties immediately imply the lemma .

	Let $u,v\in V$ and $u',v'\in V'$ be two arbitrary pairs of vertices. Let $E=2^V\cup 2^{V'}$ be the union of the complete hypergraphs on $V$ and $V'$. We define $E_1$ as $E$ along with all possible edges of size two among $\{u,v,u',v'\}$. We define $E_2$ as $E$ along with all possible edges of size $3$ among $\{u, v, u', v'\}$. It is easy to verify that for any cut, the number of edges in $E_1$ crossing the cut equals the number of edges in $E_2$ crossing the cut. So any cut size query  $\OVA$ has the same answer in $H_1$ and $H_2$, and hence can not distinguish between these two graphs, no matter the number of queries allowed.

The algorithm can additionally make at most $f(n)$ calls to $\ONONE$. But since each vertex $w \in V \cup V'$ has at least $2^n$ edges incident on it, the probability that a uniformly random edge incident on $w$ is not in $E$ is at most $3/2^n$. Using a union bound over all $f(n)$ queries along with the fact that $3f(n)/2^n \le \xi$ for sufficiently large $n$, we get that for both hypergraphs, with probability at least $1-3f(n)/2^n \ge 1-\xi$, all sampled edges are in $E$.

Thus conditioned on the event that all of the sampled edges are in $E$, the algorithm cannot distinguish between $H_1$ and $H_2$. On the other hand, there are no common edges crossing the cut $(V,V')$ in $H_1$ and $H_2$, so to output a proper $(1 \pm \eps)$-approximate cut sparsifier, the algorithm must distinguish between $H_1$ and $H_2$. Hence the probability that algorithm succeeds is at most $1/2 + \xi$.
\end{proof}

To amplify the failure probability to $1 - o(1)$, we can independently generate $\log n$ instances from the distribution above with each instance containing $n/\log n$ vertices. We now let our underlying graph be a union of these $\log n$ instances. Any algorithm that outputs a $(1\pm \eps)$-approximate sparsifier, must successfully identify for each of the $\log n$ instances whether it is an instance of $H_1$ or $H_2$. Thus the probability of success is at most $(1/2+o(1))^{\log n} = o(1)$. This completes the proof of Theorem~\ref{thm:OVA_LB}.

\subsection{$\OED$ Queries Alone are not Sufficient}
In this section, we prove that if the algorithm can access the hypergraph through only $\OED$ queries, it is not possible to find a proper $(1 \pm \eps)$-approximate cut sparsifier with $\poly(n)$ queries with success probability better than $o(1)$, proving Theorem~\ref{thm:OED_LB}.
As above, we start by showing a weaker result, which shows that the failure probability of a $\poly(n)$ time algorithm must be at least $1/2 - o(1)$, and then show how to amplify the failure probability to $1 - o(1)$.

We first define two distributions of hypergraphs $\mathcal{H}_1$ and $\mathcal{H}_2$ such that for any sequence of the queries the algorithm asks to $\OED$, the distribution of the answers are almost identical regardless of whether the graph was chosen from $\mathcal{H}_1$ or $\mathcal{H}_2$.

A graph in each of the distributions $\mathcal{H}_1$ and $\mathcal{H}_2$  is generated as follows.
There are $n+1$ vertices $v_0,v_1,\dots,v_n$ and the generated graph will have $2^n-n-1$ edges. If the graph is generated by $\mathcal{H}_1$, then we randomly choose $2^{n/2}$ subsets of $\{v_1,\dots,v_n\}$ with size at least 2. If the graph is generated by $\mathcal{H}_2$, then we randomly choose $2^{n/4}$ subsets of $\{v_1,\dots,v_n\}$ with size at least 2. Then for any subset $S$ of $\{v_1,\dots,v_n\}$ of size at least $2$, if $S$ is chosen in the previous step, then the edge $S \cup \{v_0\}$ is in the graph, otherwise the edge $S$ is in the graph.

The algorithm is presented with probability $1/2$ a graph $H$ generated by $\mathcal{H}_1$, and with probability $1/2$ a graph $H$ generated by $\mathcal{H}_2$, that is, the algorithm sees a graph $H$ generated by the distribution $1/2 \mathcal{H}_1 + 1/2\mathcal{H}_2$.
Since the cut sizes of $(\{v_0\},\overline{\{v_0\}})$ in the graph generated by $\mathcal{H}_1$ and $\mathcal{H}_2$ are $2^{n/2}$ and $2^{n/4}$ respectively, any algorithm that outputs a $(1 \pm \eps)$-approximate cut sparsifier with $\eps < 1$ must be able to distinguish between the graphs generated by $\mathcal{H}_1$ and $\mathcal{H}_2$. 
However, the following lemma shows that unless the algorithm makes exponential number of queries, it cannot distinguish between the graphs generated by $\mathcal{H}_1$ and $\mathcal{H}_2$. 

\begin{lemma} \label{lem:OED}
    Any algorithm that only makes $k$ $\OED$ queries where $k={\rm Poly}(n)$ cannot determine with probability better than $\frac{1}{2}+\frac{k^2}{2^{n/4}}$ if the underlying graph $H$ is generated from $\mathcal{H}_1$ or $\mathcal{H}_2$.
    \end{lemma}

Thus any algorithm that makes only $\poly(n)$ $\OED$ queries, fails with probability at least $1/2 - o(1)$.
To amplify the failure probability to $1 - o(1)$, we can as before independently generate $\log n$ instances from the distribution above with each instance containing $n/\log n$ vertices. We now let our underlying graph be a union of these $\log n$ instances. Any algorithm that outputs a $(1\pm \eps)$-approximate sparsifier, must successfully identify for each of the $\log n$ instances whether it was generated from the first distribution or the second. Thus the probability of success is at most $(1/2+o(1))^{\log n} = o(1)$. This completes the proof of Theorem~\ref{thm:OED_LB}. 

\subsection{Proof of Lemma~\ref{lem:OED}}

    Given a sequence of $k$ $\OED$ queries $Q = (C_1,\dots,C_k)$, we denote by $e_1,\dots,e_k$ the sequence of random edges that are returned by the query $\OED$. We first prove that for any $i$, if we fix the first $i-1$ answers $e_1,\dots,e_{i-1}$, then the distribution of edge $e_i$ is almost the same irrespective of whether the underlying graph was sampled from $\mathcal{H}_1$ or $\mathcal{H}_2$. In particular, if we denote these two distributions by $\mathcal{D}^i_1$ and $\mathcal{D}^i_2$, respectively, then we will show that the total variation distance $\tvd{\mathcal{D}^i_1}{\mathcal{D}^i_2} \le \frac{i}{2^{n/4}}$.

    Let $C_i=(S_i,\bar{S_i})$ be the $i_{th}$ cut on which the algorithm makes an $\OED$ query. Without loss of generality, assume that $v_0 \in S_i$. We first consider the case when $S_i \neq \{v_0\}$. In this case, we can sample a random edge by the following two steps: we first sample a random set $S \subseteq \{v_1,\dots,v_n\}$ that intersects both $S_i$ and $\bar{S_i}$, and we then return $S$ or $S \cup \{v_0\}$ depending on which edge is in the graph. We couple the random process that samples the edge in $\mathcal{D}^i_1$ and $\mathcal{D}^i_2$, so that in the first step, these two processes sample the same set $S$. If $S$ or $S \cup \{v_0\}$ is among $e_1 \dots e_{i-1}$, then the output of $\OED$ is fixed, which means the distributions are the same in both cases. If neither $S$ nor $S \cup \{v_0\}$ are among $e_1 \dots e_{i-1}$ and suppose there are $\ell$ different edges among $e_1,\dots,e_{i-1}$ and $j$ of them contain $v_0$, then the probability that the oracle returns $S \cup \{v_0\}$ is $\frac{2^{n/2}-j}{2^n-n-1-\ell}$ when the graph is sampled from $\mathcal{H}_1$, and $\frac{2^{n/4}-j}{2^n-n-1-\ell}$ when the graph is sampled from $\mathcal{H}_2$. So
    \begin{align*}
        \tvd{\mathcal{D}^i_1}{\mathcal{D}^i_2} \le \frac{2^{n/2}-j}{2^n-n-1-\ell} - \frac{2^{n/4}-j}{2^n-n-1-\ell} < \frac{1}{2^{n/4}}
    \end{align*}

    If $S_i = \{v_0\}$, then the oracle returns a random edge that includes $v_0$. Let $\ell$ be the number of different edges among $e_1,\dots,e_{i-1}$ and $j$ of them contain $v_0$. For any subset of $S \subseteq V$ which contains $v_0$ and has size at least $3$, if $S$ is an edge among $e_1\dots,e_{i-1}$, then $e_i=S$ with probability $\frac{1}{2^{n/2}}$ if the graph is sampled from $\mathcal{H}_1$, and with probability $\frac{1}{2^{n/4}}$ if the graph is sampled from $\mathcal{H}_2$. If $S \setminus \{v_0\}$ is among the edges $e_1,\dots,e_{i-1}$, then the probability that $e_i=S$ is $0$ for both cases. If neither $S$ nor $S \setminus \{v_0\}$ are among the edges $e_1,\dots,e_{i-1}$, then if the graph is generated by $\mathcal{H}_1$, the probability that $S$ is in the graph is $\frac{2^{n/2}-j}{2^n-n-1-\ell}$. If $S$ is indeed in the graph, then it gets sampled with probability $\frac{1}{2^{n/2}}$. So the probability that $e_i = S$ is 
$ \frac{1}{2^{n/2}} \cdot \frac{2^{n/2}-j}{2^n-n-1-\ell} =
\frac{2^{n/2}-j}{2^{n/2}(2^n-n-1-\ell)}$. Similarly, in the case of the graph generated by $\mathcal{H}_2$, the probability that $e_i=S$ is $\frac{2^{n/4}-j}{2^{n/4}(2^n-n-1-\ell)}$. Let $X = 2^n-n-1$, then we have 
    
    \begin{align*}
        2\tvd{\mathcal{D}^i_1}{\mathcal{D}^i_2} = &j \cdot \left( \frac{1}{2^{n/4}}-\frac{1}{2^{n/2}}\right) + (X-\ell)\cdot \left( \frac{2^{n/2}-j}{2^{n/2}(X-\ell)} - \frac{2^{n/4}-j}{2^{n/4}(X-\ell)}\right)\\
        \le &j \cdot \frac{1}{2^{n/4}} + \left(\frac{2^{n/2}-j}{2^{n/2}} - \frac{2^{n/4}-j}{2^{n/4}} \right) \\
        \le &\frac{2j}{2^{n/4}} \le \frac{2i}{2^{n/4}}.
    \end{align*}

    Now we are ready to prove the lemma. Let $e^1_i$ and $e^2_i$ be the random edges sampled by $\OED$ in the $i^{th}$ query when the graph is sampled from $\mathcal{H}_1$ and $\mathcal{H}_2$ respectively. Given a possible answer $\mathcal{A}_k=(e_1,e_2,\dots,e_k)$ to the $k$ queries, denote by $\mathcal{E}^1_i(\mathcal{A}_k)$ the event that $e^1_1=e_1,\dots, e^1_i = e_i$ and by $\mathcal{E}^2_i(\mathcal{A}_k)$ as the event that $e^2_1=e_1,\dots,e^2_i=e_i$. Then we can bound two times the total variation distance of the distributions of the answers when the graph is generated by $\mathcal{H}_1$ and $\mathcal{H}_2$ as below:
    
    \begin{align*}
        &\sum_{\mathcal{A}_k=(e_1,\dots,e_k)} \card{\prob{\mathcal{E}^1_k(\mathcal{A}_k)}-\prob{\mathcal{E}^2_k(\mathcal{A}_k)}} \\ 
        = &\sum_{\mathcal{A}_k=(e_1,\dots,e_k)} \card{\prob{e^1_1 = e_1 , \dots, e^1_k = e_k} - \prob{e^2_1=e_1 ,\dots, e^2_k=e_k}} \\
        = &\sum_{\mathcal{A}_k=(e_1,\dots,e_k)} \card{\prob{\mathcal{E}^1_{k-1}(\mathcal{A}_k)} \cdot \prob{e^1_k=e_k |\mathcal{E}^1_{k-1}(\mathcal{A}_k)} - \prob{\mathcal{E}^2_{k-1}(\mathcal{A}_k)} \cdot \prob{e^2_k=e_k | \mathcal{E}^2_{k-1}(\mathcal{A}_k)}}\\
        \le &\sum_{\mathcal{A}_k=(e_1,\dots,e_k)} \Big( \card{\prob{\mathcal{E}^1_{k-1}(\mathcal{A}_k)} - \prob{\mathcal{E}^2_{k-1}(\mathcal{A}_k)}}\cdot \prob{e^1_k=e_k |\mathcal{E}^1_{k-1}(\mathcal{A}_k)} \\
        & ~~~~~~~~~~~~~~~~~~~~~~+ \prob{\mathcal{E}^2_{k-1}(\mathcal{A}_k)} \cdot \card{\prob{e^1_k=e_k |\mathcal{E}^1_{k-1}(\mathcal{A}_k)} - \prob{e^2_k=e_k |\mathcal{E}^2_{k-1}(\mathcal{A}_k)} }\Big)\\
        = &\sum_{\mathcal{A}_{k-1}=(e_1,\dots,e_{k-1})} \Big( \card{\prob{\mathcal{E}^1_{k-1}(\mathcal{A}_{k-1})} - \prob{\mathcal{E}^2_{k-1}(\mathcal{A}_{k-1})}}\cdot \sum_e \prob{e^1_k=e |\mathcal{E}^1_{k-1}(\mathcal{A}_{k-1})} \\
        & ~~~~~~~~~~~~~~~~~~~~~~~~~~~~+ \prob{\mathcal{E}^2_{k-1}(\mathcal{A}_{k-1})} \cdot \sum_e \card{\prob{e^1_k=e |\mathcal{E}^1_{k-1}(\mathcal{A}_{k-1})} - \prob{e^2_k=e |\mathcal{E}^2_{k-1}(\mathcal{A}_{k-1})} }\Big)\\
        \le &\sum_{\mathcal{A}_{k-1}=(e_1,\dots,e_{k-1})} \Big( \card{\prob{\mathcal{E}^1_{k-1}(\mathcal{A}_{k-1})} - \prob{\mathcal{E}^2_{k-1}(\mathcal{A}_{k-1})}} + \prob{\mathcal{E}^2_{k-1}(\mathcal{A}_{k-1})} \cdot \frac{2k}{2^{n/4}} \Big)\\
        = &\sum_{\mathcal{A}_{k-1}=(e_1,\dots,e_{k-1})} \card{\prob{\mathcal{E}^1_{k-1}(\mathcal{A}_{k-1})} - \prob{\mathcal{E}^2_{k-1}(\mathcal{A}_{k-1})}} + \frac{2k}{2^{n/4}} \\
        \le & \sum_{\mathcal{A}_{k-2}=(e_1,\dots,e_{k-2})} \card{\prob{\mathcal{E}^1_{k-2}(\mathcal{A}_{k-2})} - \prob{\mathcal{E}^2_{k-2}(\mathcal{A}_{k-2})}} + \frac{2(k-1)}{2^{n/4}} +  \frac{2k}{2^{n/4}} \\
        & \dots \\
        \le & \sum_{i=1}^k \frac{2i}{2^{n/4}} \le \frac{2k^2}{2^{n/4}}
    \end{align*}
    Thus any algorithm that makes at most $k$ queries can distinguish between a graph generated from $\mathcal{H}_1$ and a graph generated from $\mathcal{H}_2$ with probability at most $\frac{1}{2}+\frac{k^2}{2^{n/4}}$.

%% file: weighted.tex
\section{Weighted Hypergraphs} \label{sec:weighted}

In this section, we describe how to extend the results of Theorem~\ref{thm:main}, Theorem~\ref{thm:neighbor}, Theorem~\ref{thm:spectral-neighbor}, and Theorem~\ref{thm:spectral-edge} to weighted hypergraphs. Given a weighted hypergraph $H = (V, E, w)$, we consider access to this graph by weighted generalizations of oracles $\OVA$, $\OED$, and $\ONTWO$: the oracle $\OVA$ now returns the weight of a cut rather than its size, and instead of sampling uniformly, the oracles $\OED$ and $\ONTWO$ sample edges with probability proportional to their weight.

We first note that all the Theorems we use to derive these results (Theorem~\ref{thm:kk}, Theorem~\ref{thm:spectral-S19}, and Theorem~\ref{thm:spectral-B19}) can be modified to work when the input hypergraph is weighted. 

\begin{lemma}\label{lem:kk-weighted}[Weighted version of Theorem~\ref{thm:kk}]
    Let $H=(V, E, w)$ be a weighted hypergraph with rank $r$, and let $\eps>0$ be an error parameter. Consider the hypergraph $H'$ obtained by sampling each hyperedge $e$ in $H$ independently with probability $p_e \ge \min\{1, w(e)\cdot\frac{3((d+2)\log n +r)}{k_e \eps^2}\}$, giving it weight $w(e)/p_e$ if included. Then with probability at least $1-O(n^{-d})$, $H'$ is a $(1 \pm \eps)$-approximate cut sparsifier of $H$, and has $O(\frac{n}{\eps^2}(r+\log n))$ hyperedges.
\end{lemma}

\begin{lemma}\label{lem:sy-weighted}[Weighted version of Theorem~\ref{thm:spectral-S19}]
Let $H=(V, E, w)$ be a weighted hypergraph, and let $\eps>0$ be an error parameter. Consider the hypergraph $H'$ obtained by sampling each hyperedge $e$ in $H$ independently with probability $p_e \ge \min\{1, w(e)\cdot\frac{Cn\log n}{\eps^2\min_{u,v\in e}w(E(\{u,v\}))}\}$, giving it weight $w(e)/p_e$ if included. Then with high probability, $H'$ is a $(1 \pm \eps)$-approximate spectral sparsifier of $H$, and has $\tilde{O}(n^3/\eps^2)$ hyperedges.
\end{lemma}

\begin{lemma}\label{lem:bst-weighted}[Weighted version of Theorem~\ref{thm:spectral-B19}]
    Let $H=(V, E, w)$ be a weighted hypergraph with rank $r$, and let $\eps>0$ be an error parameter. Consider the hypergraph $H'$ obtained by sampling each hyperedge $e$ in $H$ independently with probability $p_e \ge \min\{1, w(e)\cdot\frac{Cr^4r_e\log n}{\eps^2}\}$, giving it weight $w(e)/p_e$ if included. Then with high probability, $H'$ is a $(1 \pm \eps)$-approximate spectral sparsifier of $H$, and has $\tilde{O}(nr^3/\eps^2)$ hyperedges.
\end{lemma}

Most of our arguments and definitions for the unweighted case of Theorem~\ref{thm:main} and Theorem~\ref{thm:spectral-edge} translate directly to the weighted case once we replace every mention of the cardinality of an edge set by the weight of that edge set. In particular, if we generalize the definition of pseudo cut size to be $\Delta_X(S) = \frac{1}{2}(w(\delta(S))+w(\delta(X \setminus S))-w(\delta(X)))$, then the proofs for Lemma~\ref{lem:pes-sub}, Lemma~\ref{lem:pes-bound}, and Lemma~\ref{lem:sample} are completely analogous. The only difference in the analysis of Algorithm~\ref{alg:pesudo} is that the probability that a sample from the cut $(S_0, \bar{S_0})$ returns an edge $e$ in the cut is now $w(e)/w(\delta(S))$, implying that the probability that $e$ is sampled by Algorithm~\ref{alg:pesudo} is at least $\min\{1, \frac{10w(e)n^3}{\epsilon^2w(\delta(S_0))}\}$. Since this is at least the requisite sampling probability in Lemma~\ref{lem:kk-weighted}, the hypergraph $H_1$ is a sparsifier of $H$ with high probability. In the case of Theorem~\ref{thm:spectral-edge}, we still have that for every hyperedge $e$, $\frac{n}{k_e}\geq r_e$, so Algorithm~\ref{alg:pesudo} can be applied to sample each hyperedge $e$ with the desired probability of at least $\min\{1, w(e)\cdot\frac{Cr^4r_e\log n}{\eps^2}\}$.

Similarly for Theorem~\ref{thm:neighbor} and Theorem ~\ref{thm:spectral-neighbor}, the proof of correctness of Algorithm~\ref{alg:pair} is almost completely analogous (although this time, the algorithm outputs an estimate of the total weight of hyperedges containing both $u$ and $v$). The proof of correctness of Algorithm~\ref{alg:sample} is modified to assert that the probability of sampling an edge $e\in E(\{u, v\})$ is $w(e)/w(E(\{u, v\}))$, implying that the probability that the algorithm samples $e$ in each iteration is proportional to $(1\pm\eps)\cdot w(e)$.

Note that the running time of our algorithms are independent of the number of edges in the unweighted setting. Similarly, in the weighted setting, our running times have no dependence on the weights of the edges, and the running time and the size of sparsifier are the same as in the unweighted cases.

% We must also revisit some of the definitions made for unweighted graphs. 

% For the rest of this section, we fix the error parameter $\eps$ to be the reciprocal of an integer. Notice that our algorithms for Theorem~\ref{thm:main} and Theorem~\ref{thm:neighbor} are perfectly well defined if they are given access to the oracles for the weighted graph $H$ instead of an unweighted graph.

% WLOG, we assume that the minimum and maximum weights in $H$ are $3/\epsilon$ and $3W/\epsilon$ for some $W$. We define $H' = (V, E')$, an unweighted approximation to $H$, in the following way: for every edge $e\in E$, we add $\lfloor w(e)\rfloor$ copies of $e$ to $E'$. Since for any $e$, $w(e)\geq 3/\epsilon$, the number of copies of $e$ in $E'$ is within $(1\pm \epsilon/3)\cdot w(e)$
% For the rest of this section, we fix the error parameter $\eps$ and define $H' = (V, E')$, an unweighted approximation to $H$ in the following way: WLOG, assume that the minimum and maximum weights in $H$ are $1$ and $W$ respectively. Then 

%% file: conclusions.tex
\section{Concluding Remarks}
\label{sec:conclusions}

We presented the first sublinear time algorithms for creating a hypergraph sparsifier. Given access to a hypergraph through cut size and suitable edge sampling queries, our algorithm outputs a $(1 \pm \eps)$-approximate sparsifier with $\tilde{O}(n/\eps^2)$ hyperedges in polynomial time in $n$, independent of the number of hyperedges. We also showed that for any natural weakening of our query access assumptions, there is no $\poly(n)$ time algorithm for building a hypergraph sparsifier of $\poly(n)$ size. An intriguing question is if an information-theoretic cut sparsifier can be constructed using cut value queries alone. Cut value queries alone can not distinguish between hypergraphs which only contain edges of rank $2$ from hypergraphs which only contain edges of  rank $3$, making it impossible for them to output a proper sparsifier. But this does not rule out the possibility that a suitable data structure can be created using these queries alone that can recover the value of any cut to within a $(1 \pm \eps)$-approximation.